%% file: main.tex
\documentclass[11pt]{article}
\usepackage{cur-style}

\title{Communication Complexity of Inner Product in\\Symmetric Normed Spaces}
\author{Alexandr Andoni\thanks{Research supported in part by NSF grants (CCF-1740833, CCF-2008733).} \\
Columbia University \and
Jarosław Błasiok\thanks{This research was supported by a Junior Fellowship from the Simons Society of Fellows.}\\
Columbia University\and
Arnold Filtser\thanks{This research was supported by the Israel Science Foundation (grant No. 1042/22).}\\
Bar-Ilan University}
\begin{document}
\maketitle
\begin{abstract}
    We introduce and study the communication complexity of computing the inner product of two vectors, where the input is restricted w.r.t. a norm $N$ on the space $\bbR^n$.
    Here, Alice and Bob hold two vectors $v,u$ such that $\|v\|_N\le 1$ and $\|u\|_{N^*}\le 1$, where $N^*$ is the dual norm. The goal is to compute their inner product $\inprod{v,u}$ up to an $\eps$ additive term. The problem is denoted by $\IP_N$, and generalizes important previously studied problems, such as:
    (1) Computing the expectation $\mathbb{E}_{x\sim\mathcal{D}}[f(x)]$ when Alice holds $\mathcal{D}$ and Bob holds $f$ is equivalent to $\IP_{\ell_1}$. 
    (2) Computing $v^TAv$ where Alice has a symmetric matrix with bounded operator norm (denoted $\cal{S}_\infty$) and Bob has a vector $v$ where $\|v\|_2=1$. This problem is complete for quantum communication complexity and is equivalent to $\IP_{\cal{S}_\infty}$.
    
     We systematically study  $\IP_N$, showing the following results, near tight in most cases:
    \begin{enumerate}
    \item For any symmetric norm $N$, given $\|v\|_N\le 1$ and $\|u\|_{N^*}\le 1$ there is a randomized protocol using $\tilde{\Oh}(\varepsilon^{-6} \log n)$ bits of communication that returns a value in $\inprod{u,v}\pm\epsilon$ with probability $\frac23$ --- we will denote this by  $\cR_{\eps,1/3}(\IP_{N}) \leq \tilde{\Oh}(\varepsilon^{-6} \log n)$. In a special case where $N = \ell_p$ and $N^* = \ell_q$ for $p^{-1} + q^{-1} = 1$, we obtain an improved bound $\cR_{\varepsilon,1/3}(\IP_{\ell_p}) \leq \Oh(\varepsilon^{-2} \log n)$, nearly matching the lower bound $\cR_{\varepsilon, 1/3}(\IP_{\ell_p}) \geq \Omega(\min(n, \varepsilon^{-2}))$.
    \item One way communication complexity $\coR_{\varepsilon,\delta}(\IP_{\ell_p})\leq\Oh(\varepsilon^{-\max(2,p)}\cdot \log\frac n\eps)$,
    and a nearly matching lower bound $\coR_{\varepsilon, 1/3}(\IP_{\ell_p}) \geq \Omega(\varepsilon^{-\max(2,p)})$ for $\varepsilon^{-\max(2,p)} \ll n$. 
    \item One way communication complexity $\coR_{\varepsilon,\delta}(N)$ for a symmetric norm $N$ is governed by the distortion of the embedding $\ell_\infty^k$ into $N$. Specifically, while a small distortion embedding easily implies a lower bound $\Omega(k)$, we show that, conversely, non-existence of such an embedding implies protocol with communication $k^{\Oh(\log \log k)} \log^2 n$.
    \item For arbitrary origin symmetric convex polytope $P$, we show $\cR_{\eps,1/3}(\IP_{N})
    \le\Oh(\varepsilon^{-2} \log \xc(P))$, where $N$ is the unique norm for which $P$ is a unit ball, and $\xc(P)$ is the extension complexity of $P$ (i.e. the smallest number of inequalities describing some polytope $P'$ s.t. $P$ is projection of $P'$).
    \end{enumerate}
\end{abstract}
\input{intro.tex}
\input{prelim.tex}
\input{lower-bounds.tex}
\input{ell-p.tex}
\input{two-way.tex}
\input{sparsification.tex}
\input{extension-complexity-connection.tex}

\bibliography{bibliography}{}
\bibliographystyle{alpha}
\end{document}

%% file: intro.tex
\section{Introduction}
We introduce and study a class of communication problems parameterized by a norm $N$ on $\bbR^n$ --- we will denote those problems as $\IP_N$ (for \emph{inner product}). Here two parties (traditionally called Alice and Bob) are given two vectors in $\bbR^n$, respectively $v$ and $w$, with guarantees that $\|v\|_N  \leq 1$ and $\|w\|_{N^*} \leq 1$, where $N^*$ is the \emph{dual norm}  defined as $\|w\|_{N^*} := \sup_{\|v\|_N \leq 1} \inprod{v, w}$. Those players wish to use small amount of communication to compute $\inprod{v,w}$. We will focus exclusively on randomized, approximate protocols for $\IP_N$ --- for given $\varepsilon, \delta$ we denote by $\cR_{\varepsilon, \delta}(\IP_N)$ the smallest communication cost of a private-coin protocol computing some function $f(x,y)$ such that $|f(u,v) - \inprod{v, w}| \leq \varepsilon$ with probability $1-\delta$. Note that, by the definition of the dual norm we always have $|\inprod{v,w}| \leq \|v\|_N \|w\|_{N^*} \leq 1$, making this a natural normalization for the approximate version of the problem.

As it turns out this setting is not only very natural, but also quite powerful --- it is a common generalization of few seemingly different problems that have been studied in communication complexity on their own right, and are of tremendous importance in the field. We will discuss some of these connections in the next few paragraps. To start with a simple, yet consequential example, consider a finite universe $\mathcal{U}$, and a scenario where Alice has an arbitrary probability distribution $\mathcal{D}$ over $\mathcal{U}$ and Bob has a bounded function $f : \mathcal{U} \to [-1, 1]$, and they wish to compute $\mathbb{E}_{x \sim \mathcal{D}} f(x) \pm \varepsilon$. The problem they are solving is equivalent to $\IP_{\ell_1}$. As it turns out, there is an easy protocol for this problem using $\Oh(\frac{\log |\mathcal{U}|}{\varepsilon^2})$ communication --- it is enough for Alice to sample $\frac{1}{\varepsilon^2}$ elements from the distribution $\mathcal{D}$, and send their names to Bob.\atodo{Do we have a matching LB?}

\paragraph{Gap Hamming Distance problem and $\IP_{\ell_2}$.}
The Gap Hamming Distance problem has been introduced by Indyk and Woodruff in order to study the space complexity of the streaming \textsc{Distinct-Elements} problem \cite{IW03}. Here, Alice and Bob have two vectors $x, y\in \{\pm 1\}^n$, and they wish to distinguish whether $\inprod{x,y} \leq -\varepsilon n$ or $\inprod{x,y} \geq +\varepsilon n$.

After its introduction, this problem turned out to be extremely convenient as a reduction target --- communication complexity lower bound for the Gap Hamming Distance implies space lower bounds for numerous streaming and data structures problems --- and often optimal ones~(we refer to the monograph~\cite{Roughgarden16} for a detailed exposition of those consequences). As such, understanding the communication complexity of this problem has become very important --- a sequence of  lower bounds for followed \cite{Woodruff04, JKS08,Woodruff09,BC09, BCRVW10}, culminating in the optimal $\Omega(\min(\varepsilon^{-2},n))$ lower bound in the public-coin model by Chakrabarti and Regev~\cite{CR11} --- subsequently simplified and expanded upon~\cite{Sherestov12,Vidick11}.

Quite clearly, the Gap-Hamming-Distance problem is in fact a special case of the $\IP_{\ell_2}$ problem --- in particular the $\Omega(\varepsilon^{-2})$ lower bound for Gap Hamming Distance readily implies that $\cR_{\varepsilon, 1/3}(\IP_{\ell_2}) \geq \Omega(\varepsilon^{-2})$.

From the upper-bound perspective, the $\Oh(\varepsilon^{-2})$ protocol in the public-coin model is very simple for the Gap Hamming Distance, implying $\Oh(\varepsilon^{-2} + \log n)$ protocol in the private coin model, by Newmans theorem~\cite{Newman91} --- it is enough for both parties to uniformly subsample $\Oh(\varepsilon^{-2})$ coordinated and evaluate the inner product on this subset of coordinates. That was understood throughout the work on the Gap Hamming Distance.

In general, if $x,y$ are not restricted to the binary vectors, but only are restricted to be bounded in $\ell_2$ norm, the ideas of Johnson-Lindenstrauss random projections~\cite{vempala05} can provide a protocol with communication complexity $\Oh(\frac{\log n}{\varepsilon^2})$. Specifically, when $\|u\|_2 = \|v\|_2 = 1$ both Alice and Bob will use a common random projection on a subspace of dimension $\Oh(\frac{1}{\varepsilon^2})$ --- the inner product between their two vectors is preserved up to an additive error $\varepsilon$. In fact, as noticed by Canonne et. al. the space complexity of this protocol can be improved again to $\Oh(\varepsilon^{-2} + \log n)$ \cite{CGMS17} --- if Alice and Bob can use shared randomness, it is enough for them to draw $\Oh(\frac{1}{\varepsilon^2})$ random Gaussian vectors $g_1, g_2, \ldots g_k \in \bbR^n$ and compare how often $\sgn(\inprod{g_i, x}) = \sgn(\inprod{g_i, y})$. The additive $\Oh(\log n)$ term comes again from the Newman's theorem, in order to simulate this protocol in the private coin model.

\paragraph{Quantum Communication Complexity.}
Another impactful example of a well-studied problem that can be phrased as $\IP_N$ for some norm $N$ is a complete problem for \emph{quantum communication complexity} (interested reader can look at \cite{Aaronson13,Kremer95} for relevant definitions and discussion). Specifically, let us consider the communication problem where Alice is given a symmetric matrix $A \in \bbR^{d\times d}$ with bounded operator norm, and Bob a vector $v \in \bbR^d$ with $\|v\|_2 = 1$, and they wish to compute $v^T A v \pm \varepsilon$. With the proper definitions in hand, it is not difficult to see that this problem has a quantum protocol with complexity $\Oh(\frac{\log d}{\varepsilon^2})$, and on the other hand if one has a classical protocol for this problem, one can use it as a black-box to simulate arbitrary quantum communication protocol \cite{Kremer95}. As such, understanding its communication complexity has drawn significant attention: in particular Raz  showed that it has a protocol with complexity $\Oh(\sqrt{d})$ \cite{Raz99}, whereas Klartag and Regev \cite{KR11} proved $\Omega(d^{1/3})$ lower bound.  As it turns out, this problem can be easily seen to be equivalent to the $\IP_{\mathcal{S}_\infty}$, where $\mathcal{S}_\infty$ is an operator norm on the space of symmetric $d \times d$ matrices over $\mathbb{R}$. In short, to reduce the evaluation of the quadratic form to the $\IP_{\mathcal{S}_\infty}$, Bob given vector $v \in \bbR^n$ can construct a matrix $v v^T \in \bbR^{d\times d}$, such that $\inprod{A, vv^T} = \Tr A^T v v^T = v^T A v$. It can be shown that $\|v v^T\|_{\mathcal{S}_1} = 1$ where $\mathcal{S}_1$ is a \emph{Shatten 1-norm}, dual to the operator norm. The other reduction is also relatively simple.

\paragraph{Extension Complexity.}
In the seminal and celebrated work Yannakakis \cite{Yannakakis91} defined the \emph{extension complexity} of a convex polytope: extension complexity of $P$, denoted by $\xc(P)$ is the smallest number of inequalities needed to specify some higher-dimensional polytope $Q$, such that $P$ is a linear projection of $Q$. He showed a striking connection between extension complexity of a polytope $P$ and communication complexity measures of the so-called \emph{slack matrix} of $P$ --- opening a road to a number of beautiful results showing non-conditional lower bounds for extension complexity of various explicit polytopes of interests, using communication complexity methods (see \cite{GRW18, Roughgarden16} as well as \cite[Chapter 4.10]{CCZ14} and references therein).

If a polytope $P$ is origin-symmetric (i.e. satisfying $P = -P$), there is a unique norm (which we denote as $\|\cdot\|_P$ by a slight abuse of notation) for which $P$ is a unit ball. We observe that a short direct argument can be used to prove $R_{\varepsilon,1/3}(\IP_P) \leq \Oh(\frac{\log \xc(P)}{\varepsilon^2})$ --- norms associated with polytopes that have small extension complexity admit efficient communication protocols for the $\IP_P$ problem. More interestingly, by contrapositive, lower bounds for communication complexity of $\IP_{P}$ can be used to obtain lower bounds for extension complexity of $P$.

\subsection{Our results}

We initiate a systematic study of the $\IP_N$ problem, and in particular focus on one-way and two-way communication complexity of $\IP_N$ for symmetric norms $N$ in a private-coin model. We show that the two-way communication complexity for symmetric norm is small, $\poly(\log n, \varepsilon^{-1})$. In contrast, we show that one-way communication complexity is (nearly completely) determined by whether the norm contains up to small distortion $\ell_\infty^k$ --- the $k$-dimensional space equipped with the $\ell_\infty$ norm. Finally, we show connections to extension complexity, which holds for arbitrary (non necessarily symmetric) norms.

\paragraph{Communication Complexity of $\IP_{\ell_p}$.}
Given the importance of some of the communication problems $\IP_N$, it is natural to ask whether we can characterize the communication complexity of $\IP_N$ in terms of geometric properties of the norm $N$. This is quite an ambitious task: to our knowledge even the complexity of $\IP_{\ell_p}$ has not been studied before, except for $p \in \{1, 2, \infty\}$, and it is not intuitive at a first glance what answer should we expect for those norms.

It turns out to be quite easy to see that $\mathcal{R}_{\varepsilon, 1/3}(\IP_{\ell_1}) \leq \Oh(\frac{\log n}{\varepsilon^{2}})$ --- if the vector $v$ with $\|v\|_1 = 1$ is non-negative Alice can treat it as probability distribution over $[n]$ and sample $\frac{1}{\varepsilon^2}$ coordinates from this probability distribution. In the general case, they can apply the same protocol separately for the negative and positive part of the vector $v$.

Interestingly, the first protocol that comes to mind for the $\IP_{\ell_2}$ is using Johnson-Lindenstrauss lemma --- if Alice and Bob share public randomness, they can use a projection on random subspace of dimension $\Oh(\varepsilon^{-2})$. With some care this idea can be used to get a $\mathcal{R}^{\mathrm{pub}}_{\eps,\frac13}(\IP_{\ell_2}) \leq \Oh(\varepsilon^{-2})$ \cite{CGMS17}, and by Newmans theorem\cite{Newman91} $\cR_{\eps,\frac13}(\IP_{\ell_2}) \leq \Oh(\varepsilon^{-2} + \log n)$. Unfortunately, these two protocols for $\IP_{\ell_1}$ and $\IP_{\ell_2}$ seem to be very different --- and it is not clear if they can be used to obtain an efficient protocol for $\IP_{\ell_p}$ for general $p$. As it turns out, the first of the aforementioned protocols --- the coordinate-sampling based protocol --- can actually be generalized, leading to the following theorem.
\begin{theorem}
\label{thm:ell-p-cc}
For any  $1 \leq p \leq \infty$, we have $\cR_{\eps,\frac13}(\IP_{\ell_p}) \leq \Oh(\frac{\log n}{\varepsilon^2})$.
\end{theorem}
On the other hand a simple reduction from the \textsc{Gap-Hamming} problem can be used to show an almost matching lower bound $\cR_{\varepsilon, 1/3}(\IP_{\ell_p}) \geq \Omega(\varepsilon^{-2})$ for any $\varepsilon \geq \Omega(n^{-1/2})$ (\Cref{lem:eps-lb}).

\paragraph{Communication Complexity of $\IP_{N}$ for symmetric norms.}

The next step in the quest of trying to characterize the complexity of $\IP_{N}$ in terms of the geometry of the normed space $(\bbR^n, \|\cdot\|_N)$ is to consider a wide class of \emph{symmetric norm} --- norms that are invariant under permutation of coordinates and negation of any coordinate --- for example all the $\ell_p$ norms are symmetric according to this definition.

As it turns out, Andoni et. al \cite{ANNRW17} showed that any finite dimensional normed space equipped with a symmetric norm has \emph{near-isometric embedding} into a relatively ``simple'' (or at least explicit) normed space of only polynomially higher dimension (see Section~\ref{sec:symmetric-two-way} and specifically \Cref{thm:symmetric-embedding} for details).

\begin{definition}
We say that a normed $(X, \|\cdot\|_X)$ \emph{embeds into} $(Y, \|\cdot\|_Y)$ with distortion $D$ if and only if there is a linear mapping $i : X \to Y$, such that
\begin{equation*}
    \forall x \in X \qquad\|i(x)\|_Y \leq \|x\|_X \leq D \|i(x)\|_Y.
\end{equation*}
We will use notation $X \hookrightarrow^D Y$ to signify that $X$ embeds in $Y$ with distortion $D$.
\end{definition}
Perhaps surprisingly, it is not immediately obvious that such an embedding is useful for designing an efficient protocol for our problem $\IP_N$ --- in contrast with the Approximate Nearest Neighbours Problem for which this technique was developed. Specifically, let us consider a norm $N$ on $\bbR^n$, such that $\IP_N$ problem already has an efficient protocol, and a linear subspace $U \subset \bbR^n$ with an induced norm $N|_U$ --- restriction of the norm $N$ to the subspace $U$. Can we use the protocol for $\IP_{N}$ to obtain an efficient protocol for $\IP_{N|_U}$? Alice, with a vector $u \in U$ can treat it as a vector in $\bbR^n$, but what should Bob do with his vector in $U^*$ (equipped with the norm dual to $N|_U$)? This dual space cannot be interpreted in a natural way as a subspace of $(\bbR^n, N^*)$. As it turns out, $U^*$ is actually isometric with the quotient $U^* = \bbR^n / U^{\perp}$ (where $U^\perp := \{v : \forall u \in U, \inprod{u,v} = 0\}$), and importantly the unit ball of a dual norm $(N|_{U})^*$ is just an image of the ball for the norm $N^*$ under the projection $\pi : \bbR^n \to \bbR^n / U^{\perp}$ --- this is essentially a finite dimensional Hahn-Banach theorem (or equivalently separating hyperplane theorem). We can use this fact to show that $\cR_{\varepsilon, \delta}(\IP_{N|_U}) \leq \cR_{\varepsilon, \delta}(\IP_{N})$ --- Alice can just treat her vector $v \in U$ as a vector in $\bbR^n$ equipped with norm $\|\cdot\|_N$, whereas Bob given a vector $w \in \bbR^n / U^{\perp}$ can find a vector $\tilde{w} \in \pi^{-1}(w)$ such that $\|\tilde{w}\|_{N^*} = \|w\|_{(N|_U)^*}$. Using the protocol for $\IP_N$ they can approximate $\inprod{v, \tilde{w}}$, which is equal to $\inprod{v, w}$ by the choice of $\tilde{w}$.

It turns out that handling embeddings with distortion $D \geq 1$ does not provide any additional difficulty, and we can prove the following proposition.
\begin{proposition}
\label{prop:embedding-preserves-cc}
If $(\bbR^k, \|\cdot\|_X) \hookrightarrow^{\alpha} (\bbR^n, \|\cdot\|_Y)$, then for any $\varepsilon, \delta$ we have $\cR_{\varepsilon, \delta}(\IP_X) \leq \cR_{\varepsilon\alpha^{-1}, \delta}(\IP_Y)$.
\end{proposition}
With this proposition in hand, one way to show that there is an efficient protocol for $\IP_N$ for any symmetric norm $N$, is just to provide an efficient protocol for the explicit target space of the aforementioned embedding. We actually managed to implement this plan, showing the following theorem.

\begin{restatable}{theorem}{symmetricub}
\label{thm:symmetric-ub}
For any symmetric norm $N$, we have $\cR_{\varepsilon, 1/3}(\IP_N) \leq \tilde{\Oh}(\varepsilon^{-6}) \log n$.
\end{restatable}

\paragraph{One-way communication and sparsification for $\IP_N$ for symmetric norms.}
An interesting feature of the simple protocol for $\IP_{\ell_1}$ is that it is a \emph{one-way protocol}: instead of full bidirectional communication, Alice is the only party sending single message to Bob based on her vector (i.e. she sends a multiset of coordinates sampled according to the probability distribution specified by $v$), and Bob can report the answer right away, based on the message he received and his own vector $w$. We will denote by $\coR_{\varepsilon, \delta}(\IP_N)$ the one-way communication complexity for the problem $\IP_N$ up to additive error $\varepsilon$, with probability $1-\delta$.

In fact, the protocol for $\ell_1$ is a special kind of one-way protocol, which we will call \emph{sparsification} for $\IP_N$ --- Alice, based on her vector $v$, can produce a random sparse vector $\phi(v)$, such that for any fixed $w$ that Bob could have, with probability $2/3$ we have $\inprod{\phi(v), w} = \inprod{v, w} \pm \varepsilon$, and just send an encoding of this sparse vector to Bob.

In what follows let $B_N := \{x \in \bbR^n : \|x\|_N \leq 1\}$ be the unit ball for a norm $N$.
\begin{definition}
\label{def:sparsification}
We say that a norm $N$ on  $\bbR^n$ admits a $(\varepsilon, \delta, D)$-sparsification, if there exist a randomized mapping $\phi : B_N \to \bbR^n$, such that 
\begin{itemize}
    \item  For any $v \in B_N$, and any $w \in B_{N^*}$ we have $\inprod{\phi(v), w} = \inprod{v, w} \pm \varepsilon$ with probability at least $1 - \delta$.
    \item For any $v$, we always have $\|\phi(v)\|_0 \leq D$ (i.e. $\phi(v)$ has at most $D$ non-zero coordinates).\atodo{added note}
\end{itemize}
\end{definition}
It is not difficult to see that if a norm $N$ admits an $(\varepsilon, \delta, D)$-sparsification, then $\coR_{\varepsilon, \delta}(\IP_N) \leq \Oh(D\cdot \log\frac n\eps)$ (see \Cref{prop:sparsification-implies-cc}), but such a sparsification can potentially be much more useful than arbitrary one-way protocol --- in particular, note that Bob in order to compute his answer does not have to access the entire vector $w$, just $D$ of its coordinates.

As it turns out, the protocol for $\ell_p$ norms is actually achieved via such sparsification.
\begin{restatable}[]{theorem}{SparsificationLp}
	\label{thm:sparsification}
	For any $1 \leq p < \infty$, and any $\varepsilon>0$, the norm $\ell_p$ on $\bbR^n$ admits a $(\varepsilon, 1/3, \Oh(\varepsilon^{-\max(2, p)}))$-sparsification.
	In particular, $\coR_{\varepsilon,\delta}(\IP_{\ell_p})\leq\Oh(\varepsilon^{-\max(2,p)}\cdot\log\frac n\eps)$.\atodo{conclusion added}
\end{restatable}
Interestingly, the guaranteed sparsity in the theorem above does not depend on the ambient dimension $n$; for a fixed $p$ and $\varepsilon$, we could approximate $\inprod{v, w}$ by accessing only constant number of coordinates of $w$.
\atodoin{There is inconsistency with the notation for dimension. Sometimes it is $d$ and sometimes it is $n$.}

In fact it is not difficult to show that this sparsification result is almost optimal. A direct reduction from the \textsc{index} problem can be used to show $\coR_{\varepsilon, 1/3}(\ell_\infty) \geq \Omega(n)$ for any $\varepsilon < 1$. Since Proposition~\ref{prop:embedding-preserves-cc} has a direct analog for one-way communication complexity (with the same proof), we have as a consequence
\begin{proposition}
\label{prop:embedding-lb}
For any norm $N$ on $\bbR^n$, if $(\bbR^k, \ell_\infty) \hookrightarrow^{\varepsilon^{-1}} (\bbR^n, N)$, then $\coR_{\varepsilon, \delta}(\IP_N) \geq \Omega(k)$.
\end{proposition}
We use this, together with known embeddings of low-dimensional $\ell_\infty$ into $\ell_p$, to deduce lower bounds $\coR_{\varepsilon, 2/3}(\IP_{\ell_p}) \geq \Omega(\varepsilon^{-\max(2, p)})$ --- nearly matching (up to the $\log n$ multiplicative factor) the upper bound from Theorem \ref{thm:sparsification} obtained via sparsification.

Perhaps surprisingly, among all symmetric norms $N$, embedding of $\ell_\infty^k$ with small distortion is the only obstruction for a good sparsification of a norm --- we prove the following weak converse of Proposition~\ref{prop:embedding-lb}.
\begin{restatable}{theorem}{symmetricsparsification}
\label{thm:symmetric-sparsification}
If $N$ is a symmetric norm such that $(\bbR^k, \ell_\infty) \not\hookrightarrow^{\varepsilon^{-1}} (\bbR^n, N)$, then $N$ admits a $(\Oh(\varepsilon), 1/3, k^{\Oh(\log \log k)} \log^2 n)$-sparsification.\atodo{Can we apply this on some norm to get something otherwise unknown?}
\end{restatable}

This theorem is in fact the most technically challenging part of this paper, as well as most unexpected --- while Proposition~\ref{prop:embedding-lb} says that existence of an embedding implies a communication complexity lower bound, which is relatively common phenomenon, here we claim to be able to provide an efficient protocol given \emph{non-existence} of any embedding from one space to another (or, by contrapositive, claim an existence of embedding assuming non-existence of efficient protocol). 

Theorem~\ref{thm:symmetric-sparsification} together with Proposition~\ref{prop:embedding-lb} essentially say that the one-way communication complexity of any symmetric norm is characterized by the possibility of finding low-distortion embedding of finite-dimensional $\ell_\infty$ spaces into this norm.

We believe that both the weak dependency on $\log n$ and the unfortunate dependency on $k$ are just artifacts of our analysis. In fact, we conjecture that
\begin{conjecture}
If $N$ is a symmetric norm such that $(\bbR^k, \ell_\infty) \not\hookrightarrow^{\varepsilon^{-1}} (\bbR^n, N)$, then $N$ admits an $(\Oh(\varepsilon), 1/3, k^{\Oh(1)})$-sparsification.
\end{conjecture}
\paragraph{Connections with extension complexity.}
For any origin symmetric convex polytope $P$ (i.e. satisfying $P = -P$), we can define the associated norm $\|w\|_P := \inf \{ \lambda : w / \lambda \in P \}$ --- this is a unique norm for which $P$ is a unit ball. As mentioned earlier, it is easy to directly show the following relation between extension complexity of $P$ and communication complexity of $\IP_P$.
\begin{proposition}
\label{prop:xc-weak}
For any origin symmetric convex polytope $P$ we have $\cR_{\varepsilon, 1/3}(\IP_P) \leq \Oh(\frac{\log \xc(P)}{\varepsilon^2})$.
\end{proposition}

We can sketch the proof in a single paragraph --- first by dual of the Proposition~\ref{prop:embedding-preserves-cc}, we show that if $P$ is a projection of a polytope $Q$, then $\cR_{\varepsilon, 1/3}(\IP_P) \leq \cR_{\varepsilon, 1/3}(\IP_Q)$, by symmetry $\cR_{\varepsilon, 1/3}(\IP_{Q}) = \cR_{\varepsilon, 1/3}(\IP_{Q^*})$, and finally with a sampling based argument we can show that if the unit ball for a norm $\|\cdot\|_{Q^*}$ has at most $M$ vertices, then $\cR_{\varepsilon, 1/3}(\IP_{Q^*}) \leq \Oh(\frac{\log M}{\varepsilon^2})$ --- this step is essentially the same as the proof that $\cR_{\varepsilon, 1/3}(\IP_{\ell_1}) \leq \Oh(\frac{\log n}{\varepsilon^2})$. Finally vertices of the dual polytope $Q^*$ correspond exactly to facets of $Q$.

For a convex polytope $P \subset \bbR^n$, specified by matrix $A \in \bbR^{M \times n}$ via $P = \{ v : \forall i\leq M, \inprod{A_i, v} \leq b_i\}$, we define the \emph{slack matrix} $S_P \in \bbR^{|V| \times M}$ with rows indexed by vertices $V$ of the polytope $P$, to be $(S_P)_{v, i} := b_i - \inprod{A_i, v}$ --- clearly this matrix is non-negative. Yannakakis showed that $\xc(P) = \rk^{+}(S_P) + \Oh(1)$, where $\rk^{+}(S)$ is a \emph{non-negative rank} of $S$ --- i.e. the smallest number $k$ such that we can decompose $S = \sum_{i \leq k} v_i w_i^T$ for coordinatewise non-negative vectors $v_i, w_i$ \cite{Yannakakis91}. Therefore, Proposition~\ref{prop:xc-weak} can be equivalently stated as bounding the communication complexity of $\IP_P$ by the non-negative rank of the corresponding slack matrix, i.e. $\cR_{\varepsilon, 1/3}(\IP_P) \leq \Oh(\frac{\log \rk^{+}(S_P)}{\varepsilon^2})$.

It turns out that $\cR_{\varepsilon, 1/3}(\IP_P)$ is in turn lower bounded (up to the dependency on $\varepsilon$) by a closely related quantity, the \emph{approximate non-negative rank} of the same matrix $S_P$.
\begin{definition}{\cite{KMSY14}}
For a non-negative matrix $S \in \bbR^{n\times m}$ and $\varepsilon \geq 0$, we define \emph{approximate non-negative rank} as
\begin{equation*}
    \annr{\varepsilon}(S) := \inf_{S' : \forall_{ij} |S_{ij} - S'_{ij}| \leq \varepsilon} \rk^{+}(S').
\end{equation*}
\end{definition}
Clearly $\annr{\varepsilon}(S) \leq \rk^{+}(S)$. We will show in Section~\ref{sec:extension-complexity} that this former quantity turns out to lower bound the communication complexity of $\IP_P$.
\begin{proposition}
\label{thm:xc-strong}
For any origin-symmetric convex polytope $P$ we have
\begin{equation*}
    \annr{2\varepsilon}(S_P) \leq \cR_{\varepsilon, \varepsilon}(\IP_P) \leq \Oh(\log \varepsilon^{-1} \cR_{\varepsilon, 1/3}(\IP_P)).
\end{equation*}
\end{proposition}

%% file: prelim.tex
\section{Preliminaries and notation}
\atodoin{We should define communication complexity...}
All the results and notations in this section are standard.

In the sequel we write $w = v \pm \varepsilon$ to indicate that $|w - v| \leq \varepsilon$. We will use notation $A \lesssim B$ to denote that there exist a universal constant $K$ such that $A \leq K B$. 

For a norm $N$, we write $B_N := \{ v : \|v\|_N \leq 1\}$ to denote the unit ball with respect to norm $N$, and if $N = \ell_p$ we usually abbreviate $B_p$ instead of $B_{\ell_p}$. 

We will also use $\tilde{O}(\cdot)$ notation to supress polylogartihmic factors, i.e. $g = \tilde{\Oh}(f)$ if and only if $g = \Oh(f \log^{\Oh(1)} f)$.
\subsection{Symmetric norms}
The following notion of symmetric norm will be used throughout the paper.
\begin{definition}
A norm $N$ on $\bbR^n$ is a \emph{symmetric norm} if and only if it satisfies the following two properties:
\begin{enumerate}
    \item Invariant under permutations: $\|v\|_N = \|\pi(v)\|_N$ for any permutation $\pi : [n] \to [n]$, where $\pi(v)$ is a vector with $\pi(v)_i = v_{\pi(i)}$.
    \item Invariant under coordinate negation: $\|v\|_N = \| |v| \|_N$ where $|v|$ is a vector with $|v|_i = |v_i|$.
\end{enumerate}
We will deal only with symmetric norms normalized as $\|e_1\|_N = 1$ for a standard basis vector $e_1$, and this normalization will not be mentioned explicitly anymore.\atodo{You say that all notations are standard. Is this as well? otherwise it is sub-optimal that it is hidden here.}
\end{definition}
\begin{claim}
    Any symmetric norm is coordinate-wise monotone: if vectors $v$ and $w$ satisfy $\forall_i |v_i| \leq |w_i|$, then $\|v\|_N \leq \|w\|_N$.
\end{claim}
\begin{proof}
It is enough to show this inequality in the special case where $v_i = w_i$ for $i \geq 2$, and $|v_1| \leq |w_1|$. In this case, we can define $w' \in \bbR^n$ as $w'_1 := -w_1$, and $w'_i := w_i$ for $i\geq 2$. Clearly, $\|w'\|_N = \|w\|_N$, and moreover $v$ is a convex combination of $w$ and $w'$: $v = \lambda w + (1-\lambda) w'$ where $\lambda = \frac{1}{2} + \frac{|v_1|}{2 |w_1|}$. Hence, by triangle inequality $\|v\|_N \leq \lambda \|w\|_N + (1 - \lambda) \|w'\|_N = \|w\|_N$
\end{proof}

\subsection{Probabilistic Tools}
For an $\bbR$-valued random variable $Z$ and $p\geq 1$, we will consider the $p$-norm $\|Z\|_p := (\mathbb{E} |Z|^p)^{1/p}$, and the weak $p$-norm:
\begin{equation*}
    \|Z\|_{p, \infty} := \sup_{\lambda} \Pr(|Z| \geq \lambda)^{1/p} \lambda.
\end{equation*}
\begin{fact}[Monotonicity of $L_p$]
For $1 \leq p \leq p' \leq \infty$, we have $\|Z\|_{p} \leq \|Z\|_{p'}$ and $\|Z\|_{p,\infty} \leq \|Z\|_{p', \infty}$.
\end{fact}
Markov inequality says that $\|Z\|_{p, \infty} \leq \|Z\|_{p}$, and this can be weakly inverted:
\begin{lemma}
\label{lem:weak-markov-inverse}
For any $1 \leq p' < p$, we have
\begin{equation*}
    \|Z\|_{p'} \leq \left(\frac{p}{p - p'}\right)^{1/p'} \|Z\|_{p, \infty}
\end{equation*}
\end{lemma}
\begin{proof}
We assume that $\|Z\|_{p, \infty} < \infty$, because otherwise the inequality is trivial.


Integration by parts yields
\begin{align*}
    \E |Z|^{p'} & = - \lim_{\lambda \to \infty} \Pr(|Z| \geq \lambda) \lambda^{p'} +  p' \int_0^\infty \lambda^{p'-1} \Pr(|Z| \geq \lambda) \d\lambda.
\end{align*}
We can bound $\Pr(|Z| \geq \lambda) \lambda^{p'} \leq \|Z\|_{p, \infty}^p \lambda^{p' - p} \to^{\lambda \to \infty} 0$.

For the second term
\begin{align*}
    p' \int_0^\infty \lambda^{p'-1} \Pr(|Z| \geq \lambda) \d\lambda & \leq p' \int_{0}^{\infty} \lambda^{p' - 1} \min(1, \|Z\|_{p, \infty}^p \lambda^{-p}) \d \lambda \\
    & \leq p' \left( \int_0^{\|Z\|_{p, \infty}} \lambda^{p' - 1} \d \lambda + \|Z\|_{p, \infty}^p \int_{\|Z\|_{p, \infty}}^{\infty} \lambda^{p' - 1 - p} \d \lambda \right) \\
    & \leq p' \left( \frac{1}{p'} \|Z\|_{p, \infty}^{p'} + \frac{1}{p - p'} \|Z\|_{p, \infty}^{p'}\right) \\
    & \leq \frac{p}{p - p'} \|Z\|_{p, \infty}^{p'}.
\end{align*}
%
\end{proof}

The following lemma is also standard, we include the proof for completeness.

\begin{lemma}
		\label{lem:type}
		Let $1 \leq q \leq 2$, and $Z_1, Z_2, \ldots Z_s$ be independent random variables in $\bbR$ satisfying $\E Z_i = 0$. Then
		\begin{equation*}
			\|\sum_i Z_i\|_q \lesssim (\sum_i \|Z_i\|_q^q)^{1/q}.
		\end{equation*}
		In particular, if $Z_i$ are i.i.d., we have $\|\sum_{i\leq s} Z_i\|_q \lesssim s^{1/q} \|Z_1\|_q$.
\end{lemma}
\begin{proof}
		Let $Z'_i$ be all independent, and $Z'_i$ distributed identically as $Z_i$, and let $\varepsilon_i$ be again independent Rademacher random variables. We have
		\begin{equation*}
				\|\sum_i Z_i\|_q \overset{(1)}{=} \|\sum_i Z_i - \E Z_i\|_q \overset{(2)}{\leq} \|\sum_i Z_i - Z'_i\|_q \overset{(3)}{=} \|\sum_i \varepsilon_i (Z_i - Z'_i)\|_q \overset{(4)}{\leq} 2 \|\sum_i\varepsilon_i Z_i\|_q.
		\end{equation*}
		The $(1)$ equality follows as $Z_i$ are mean $0$.
		For the $(2)$ inequality, since $q \geq 1$, the function $x^q$ is convex, we can pull out the inner expectation (denote $Z = \sum Z_i$, and $Z' = \sum Z_i'$, then by Jensen we have $\|Z-\mathbb{E}_{Z'}Z'\|_{q}=(\mathbb{E}_{Z}(\mathbb{E}_{Z'}[Z-Z'])^{q})^{1/q}\le(\mathbb{E}_{Z}\mathbb{E}_{Z'}[(Z-Z')^{q}]))^{1/q}=\|Z-\mathbb{E}_{Z'}Z'\|_{q}$). 
		For equality $(3)$, note that since $Z_i$ and $Z_i'$ have the same distribution and are independent, then distribution of $Z_i - Z_i'$ is symmetric. I.e. $Z_i - Z_i'$ has the same distribution as $-(Z_i - Z_i') = Z_i' - Z_i$. In particular, 
		also $\varepsilon_i (Z_i - Z_i')$ has also the same distribution.
		Finally inequality $(4)$ is just a triangle inequality.
		
		The bound $\|\sum \varepsilon_i Z_i \|_q \leq (\sum \|Z_i\|_q^q)^{1/q}$ is just the type-$q$ inequality for the norm $L_q$. To see why it is true, we can condition on $Z_i$ and apply the Khintchine inequality (which says that for fixed numbers $a_1, \ldots a_n$ and independent Rademacher random variables $\varepsilon_i$, we have $\E \inprod{\varepsilon, a}^q \lesssim (\sum a_i^2)^{q/2})$ to get
		\begin{equation*}
		    (\E \inprod{\varepsilon, Z}^q)^{1/q} = (\E_Z \E_\varepsilon \inprod{\varepsilon, Z}^q)^{1/q} \lesssim (\E_Z (\sum Z_i^2)^{q/2})^{1/q},
		\end{equation*}
		Since $q \leq 2$, we have $(\sum Z_i^2)^{q/2} \leq \sum Z_i^q$, hence
		\begin{equation*}
		    \|\sum_i\varepsilon_i Z_i\|_q \lesssim (\sum_i \E Z_i^q)^{1/q}.
		\end{equation*}
\end{proof}

\subsection{Duality}
\begin{lemma}[Finite dimensional Hahn-Banach theorem]
\label{lem:hahn-banach}
Consider a normed space $(\bbR^n, \| \cdot \|_N)$, and a subspace $V \subset \bbR^n$ with the induced norm given by the restriction of $\|\cdot\|_N$ to this subspace.

Then the dual $V^*$ is isometric with the quotient $\bbR^n / V^\perp$, with the norm given by $\|w + V^\perp\|_{V^*} = \min_{v^{\perp}\in V^{\perp}} \|w + v^{\perp}\|_{N^*}$.

In other words: for any vector $w \in \bbR^n$, we have $\max_{v \in V \cap B_N} \inprod{v, w} = \min_{v^\perp \in V^{\perp}} \|w + v^\perp\|_{N^*}$.
\end{lemma}
\begin{proof}
As usual in the duality statement, one inequality is easy: clearly, for any $v \in V \cap B_N$, and any $v^\perp \in V^{\perp}$, we have
\begin{equation*}
    \inprod{v,w} = \inprod{v,w + v^{\perp}} \leq \|v\|_N \|w + v^{\perp}\|_{N^*} \leq \|w + v^{\perp}\|_{N^*},
\end{equation*}
so
\begin{equation*}
    \max_{v \in V \cap B_N} \inprod{v, w} \leq \min_{v^{\perp} \in V^{\perp}} \|w + v^{\perp}\|_{N^*}.
\end{equation*}

For the other direction, consider a projection $\pi : \bbR^n \to \bbR^{n} / V^{\perp}$. Take a vector $w$ such that $\max_{v \in V \cap B_N} \inprod{v,w} = 1$, and assume for the sake of reaching a contradiction that $\min_{v^{\perp} \in V^{\perp}} \|w + v^{\perp}\|_{N^*} > 1$. Equivalently, this is saying that $\pi(w) \not \in \pi(B_{N^*})$, hence by a separating hyperplane theorem we can find a hyperplane separating $\pi(w)$ from $\pi(B_{N^*})$. Such a hyperplane will be parametrized by $r \in V$, say with $\|r\|_N = 1$, s.t. $\inprod{r,w} > \max_{w' \in B_{N^*}} \inprod{v, w'} = \|r\|_{N} = 1$. But this is a contradiction with $\max_{v \in V \cap B_N} \inprod{v, w} = 1$.
\end{proof}

The following is basically a restatement of the more involved direction in the above duality.
\begin{claim}
\label{claim:duality}
If $p: X \to Y$ is an isometric embedding, then for any $w \in X^*$ we can find $w' \in Y^*$, with $\|w'\|_{Y^*} = \|w\|_{X^*}$, s.t. for any $x \in X$ we have $\inprod{p(x), w'} = \inprod{x, w}$.
\end{claim}
\begin{remark}
Note that the map which sends $w \in X^*$ to $w' \in Y^*$ described above is usually not linear.
\end{remark}

For a convex, origin-symmetric body $K \subset V$, we define a dual body $K^* \subset V^*$ as $K^* := \{ w : \forall v\in K,\, \inprod{w, v} \leq 1\}$. If $N$ is a norm, we have $B_N^* = B_{N^*}$, moreover for any pair of bodies we have $K_1 \subset K_2 \iff K_2^* \subset K_1^*$, and $K^{**} = K$. Finally, for arbitrary subset $S\subset V$ define $\conv(S)$ to be the convex hull of $S$.

\begin{lemma}
\label{lem:dual-intersection}
If $K_1$ and $K_2$ are convex origin-symmetric bodies, then
\begin{equation*}
    (K_1 \cap K_2)^* = \conv(K_1^* \cup K_2^*).
\end{equation*}
\end{lemma}
\begin{proof}
    We will first show that $(K_1 \cap K_2)^* \subset \conv(K_1^* \cup K_2^*)$. Equivalently we can check that $\conv(K_1^* \cup K_2^*)^* \subset K_1 \cap K_2$. To this end, take arbitrary $v \in \conv(K_1^* \cup K_2^*)^*$. By definition, for any $w \in \conv(K_1^* \cup K_2^*)$ we have $\inprod{v, w} \leq 1$. In particular since $K_1^* \subset \conv(K_1^* \cup K_2^*)$, we have $\sup_{w \in K_1^*} \inprod{v, w} \leq 1$, therefore $v \in K_1^{**} = K_1$, and symmetrically $v \in K_2$.
    
    For the other direction, we wish to show that $\conv(K_1^* \cup K_2^*) \subset (K_1 \cap K_2)^*$. Note that since $K_1 \cap K_2 \subset K_1$, we have $K_1^* \subset (K_1 \cap K_2)^*$, and symmetrically $K_2^* \subset (K_1 \cap K_2)^*$.
    
    Since $(K_1 \cap K_2)^*$ is a convex body such that $K_1^* \cup K_2^* \subset (K_1 \cap K_2)^*$, we also have $\conv(K_1^* \cup K_2^*) \subset (K_1 \cap K_2)^*$.
\end{proof}

\begin{corollary}
\label{cor:decomposition}
Consider two arbitrary norms $N_1, N_2$ on a linear space $V$, and a norm $N = \max(N_1, N_2)$. We can decompose any $w \in B_{N^*}$ as $w = w_1 + w_2$, where $\|w_1\|_{N_1^*} + \|w_2\|_{N_2^*} \leq 1$. 
\end{corollary}

%% file: lower-bounds.tex
\section{Reductions and communication complexity lower bounds}
In this section we prove that a bounded distortion embedding of a normed space $(\bbR^k, \|\cdot\|_X)$ into another normed space $(\bbR^n, \|\cdot\|_Y)$ provides a reduction between the corresponding $\IP_X$ and $\IP_Y$ communication problems. We use this reduction to deduce lower bounds for one-way communication complexity of the $\IP_X$ problems for spaces $X$ that contain a (distorted) copy of $\ell_{\infty}$.

\begin{proposition}[Extended version of Proposition~\ref{prop:embedding-preserves-cc}]
\label{prop:embedding-preserves-cc-ext}
If $(\bbR^k, \|\cdot\|_X) \hookrightarrow^{\alpha} (\bbR^n, \|\cdot\|_Y)$, then for any $\eps, \delta$ we have $\cR_{\varepsilon, \delta}(\IP_X) \leq \cR_{\alpha^{-1} \varepsilon, \delta}(\IP_Y)$, and moreover $\coR_{\varepsilon, \delta}(\IP_X) \leq \coR_{\alpha^{-1}\varepsilon, \delta}(\IP_Y)$.
\end{proposition}
\begin{proof}
We can decompose any embedding $f : X \to Y$ with distortion $\alpha$ as a composition of linear maps $f = f_2 \circ f_1$, where $f_1 : X \to f(X)$ is an isomorphism of underlying vector spaces (it preserves the dimension) and has distortion $\alpha$ and $f_2 : f(X) \to Y$ is an isometric embedding (i.e. it has distortion $1$).

As such it is enough to show the statement of the proposition in those two special cases --- if the  distortion is $1$, or map $f$ is in an isomorphism (perhaps with non-trivial distortion).\atodo{Not very clear, in the second case you mean that $f_2$ is the identity?}

For isometric embedding $p : X \to Y$ the result $\cR_{\varepsilon, \delta}(\IP_X) \leq \cR_{\varepsilon, \delta}(\IP_Y)$ (and the same for $\coR$) follows from \Cref{claim:duality} --- Alice, given a vector $u \in X$, can instead consider a vector $p(u) \in Y$, Bob on the other hand given $w \in X^*$ can find a vector $w' \in Y^*$ as in \Cref{claim:duality} --- such that $\inprod{u,w} = \inprod{p(u), w'}$. They can now apply the protocol for $\IP_Y$ on the pair of inputs $(p(u), w')$.

Now, consider the assumption that $f$ is an isomorphism of underlying vector spaces, with distortion $\alpha$. In this case we can basically think of a single vector space $V$, on which we have two norms $X$ and $Y$, s.t. for all $v$ we have $\alpha^{-1} \|v\|_X \leq \|v\|_Y \leq \|v\|_X$. It is easy to check that in this case, for any $w \in V^*$ we have $\|w\|_{X^*} \leq \|v\|_{Y^*} \leq \alpha \|v\|_{X^*}$.

Finally, if Alice has $v \in B_X$ (unit ball with respect to the norm $\|\cdot\|_X$), and Bob has $w \in B_{X^*}$, we have $v \in B_Y$ and $\alpha^{-1} w \in B_{Y^*}$. They can use the protocol for the norm $\|\cdot\|_Y$ with an error $\alpha^{-1} \varepsilon$ applied to $(v, \alpha^{-1} w)$ to compute jointly $C = \alpha^{-1} \inprod{v,w} \pm \alpha^{-1} \varepsilon$. Knowing this $C$ they can easily compute $\alpha C = \inprod{x, y} \pm \varepsilon$.
\end{proof}
\begin{lemma}
\label{lem:l-infty-lb}
    For any $\alpha< 1$, we have $\coR_{\alpha, 1/3}(\IP_{\ell_{\infty}}) \geq \Omega(n).$
\end{lemma}
\begin{proof}
We will use reduction from the \textsc{index} problem on $n$ bits --- in this problem Alice is given an $n$ bit string $x \in \{0, 1\}^n$, and Bob an index $i \in [n]$, and he wants to output a value $x_i$ with probability at least $2/3$. It is known that one-way communication of the \textsc{index} problem is $\Omega(n)$ --- Alice needs to essentially send her entire bit-string to Bob \cite{KNR95, Roughgarden16}.

Given an instance $(x,i)$ of an indexing problem we can easily encode it as an equivalent instance of $\IP_{\ell_{\infty}}$: Alice takes her bit-string $x \in \{0, 1\}^n$ encodes it as $v_x = \sum_i x_i e_i$, with $\|v_x\|_\infty \leq 1$. Bob can recover arbitrary bit $x_i$ if he can approximate $\inprod{v_x, e_i} = x_i$ up to an error smaller than $1$. His test vector $w=e_i$ clearly has $\|w\|_1 = 1$.
\end{proof}

We are now ready to prove Proposition~\ref{prop:embedding-lb} announced in the introduction.
\begin{proof}[Proof of Proposition~\ref{prop:embedding-lb}]
We need to show that if $\ell_{\infty}^k \hookrightarrow^{\varepsilon^{-1}} (\bbR^n, \|\cdot\|_X)$, then $\coR_{\varepsilon, 1/3}(\IP_X) \geq \Omega(k)$. It follows as a direct corollary by composing Lemma~\ref{lem:l-infty-lb} together with Proposition~\ref{prop:embedding-preserves-cc-ext}. 
\end{proof}

Finally, we can use Proposition~\ref{prop:embedding-lb} to show a lower bound for the communication complexity of $\IP_{\ell_p}$.
\begin{lemma}
We have $\coR_{\varepsilon, 1/3}(\IP_{\ell_p}) \geq \Omega(\min(n, \varepsilon^{-\max(p,2)}))$.
\end{lemma}
\begin{proof}
Let $r = \varepsilon^{-\max(p, 2)}$, and assume $r < n$. We wish to show that $\ell_\infty^r \hookrightarrow^{\varepsilon^{-1}} \ell_p$.
If $p \geq 2$, we can just take a subspace spanned by the first $r$ basis vectors, and observe that for $\alpha \in \ell_{\infty}^r$, we have $\|\alpha\|_\infty \leq \|\sum_{i \leq r} \alpha_i e_i\|_p \leq \|\alpha\|_\infty r^{1/p} = \varepsilon^{-1} \|\alpha\|_\infty$ --- this inclusion gives an embedding with distortion $\varepsilon^{-1}$.

For $p < 2$, it is known that the Banach-Mazur distance between $\ell_p^r$ and $\ell_{\infty}^r$ is at most $\Oh(\sqrt{r}) = \Oh(\varepsilon^{-1})$ \cite[Chapter 1, Section 8]{johnson2001handbook} --- which is just the same as saying that there exists a linear isomorphism between $\ell_{\infty}^r$ and $\ell_p^r$ with distortion at most $\Oh(\sqrt{r}) = \Oh(\varepsilon^{-1})$.
\end{proof}

\begin{lemma}
\label{lem:eps-lb}
We have $\cR_{\varepsilon, 1/3}(\IP_{\ell_p}) \geq \Omega(\min(n, \varepsilon^{-2}))$.
\end{lemma}
\begin{proof}
We will use a reduction from the \textsc{Gap-Hamming} problem --- in this problem Alice and Bob are given two vectors $x, y \in \{0, 1\}^k$, and they wish to check if $\Delta(x,y) < \frac{k}{2} - C \sqrt{k}$ or $\Delta(x,y) > \frac{k}{2} + C \sqrt{k}$, where $\Delta(x,y) = | \{ i : x_i \not= y_i\} |$ --- it is known that any randomized protocol solving this problem with probability $2/3$ needs to use $\Omega(k)$ communication \cite{CR11, Sherestov12}.

Alice and Bob will interpret their vectors $x, y$ directly as vectors in $\bbR^n$ (padding with zeros if necessary). Note that since $\Delta(x,y) = \|x\|_2^2 + \|y\|_2^2 - 2 \inprod{x,y}$, if Alice and Bob can compute $\inprod{x,y}$ up to an additive error smaller than $2 C \sqrt{k}$, they can solve the \textsc{Gap-Hamming} problem. Using the $\varepsilon$-approximate protocol for $\IP_{\ell_p}$ they can compute $\inprod{x,y} \pm \varepsilon\|x\|_p \|y\|_q = \inprod{x,y} \pm \varepsilon k^{1/p} k^{1/q} = \inprod{x,y} \pm \varepsilon k$. Choosing $\varepsilon = \frac{2 C}{\sqrt{k}}$, we get the desired lower bound: $\cR_{\varepsilon, 1/3}(\IP_{\ell_p}) \geq \Omega(k) = \Omega(\varepsilon^{-2})$.
\end{proof}

Finally, we will show we can solve the one-way communication problem $\IP_N$ if the norm $N$ admits a sparsification.
\begin{proposition}
\label{prop:sparsification-implies-cc}
If a norm $N$ admits a $(\varepsilon, \delta, D)$-sparsification, then $\coR_{\varepsilon, \delta}(\IP_N) \leq \Oh(D\cdot\log\frac n\eps)$.
\end{proposition}
\begin{proof}
Alice, given a vector $v \in \bbR^n$ with $\|v\|_{N} \leq 1$ can just apply a sparsification procedure to obtain a vector $\phi(v)$ and try to send an encoding of $\phi(v)$ to  Bob. In order to efficiently encode this vector, for each of at most $D$ non-zero coefficients she needs to send the index of the corresponding coefficient (spending at most $\log n$ bits), and an encoding of its value up to precision $(\varepsilon/n)^{\Oh(1)}$ --- for which she needs to spend $\Oh(\log n + \log \varepsilon^{-1})=\Oh(\log \frac n\eps)$ bits.
\end{proof}

%% file: ell-p.tex
\section{Sparsification for  $\ell_p$}
\label{sec:ell-p}
In this section we prove \Cref{thm:sparsification}, restated for convenience.
\SparsificationLp*
\begin{proof}
It is enough to show the sparsification procedure $\phi$ for vectors $v \in \bbR^n$ with $\|v\|_p = 1$ --- for arbitrary vector $v \in B_p$, we will then take $\phi(v) := \|v\|_p \phi(v / \|v\|_p)$.

Given $v$ with $\|v\|_p = 1$, consider a distribution $\mathcal{D}_v$ over $[n]$, given by $\Pr_{a \sim \mathcal{D}_v}(a = i) = |v_i|^p$.

Consider now a random vector $\tilde{v} := \frac{e_a v_a}{|v_a|^p}$ where $a \sim \mathcal{D}_v$ and $e_a$ is a standard basis vector. Clearly, we have $\E \tilde{v} = v$.

We will pick some $s = \Theta(\varepsilon^{-\max(2, p)})$, and define $\phi(v) = \frac{1}{s} \sum_{j\leq s} \tilde{v}^{(j)}$, where $\tilde{v}^{(j)}$ are i.i.d. random vectors with the same distribution as $\tilde{v}$. Again, clearly $\E \phi(v) = v$, and $\|\phi(v)\|_0 \leq s$. 

\jtodo{I don't think we need anywhere bounds on $\ell_1$ norm of $\phi(v)$?}

Consider some $w \in B_q^n$, where $q$ is such that $\frac{1}{p} + \frac{1}{q} = 1$. We wish to control the error $\inprod{\phi(v), w} - \inprod{v, w}$.
\begin{claim}
We have $\| \inprod{\tilde{v}, w}\|_q \leq 1$. (This notation is a bit confusing: here $\inprod{\tilde{v}, w}$ is an $\bbR$-valued random variable, and the $L_q$-norm under consideration is $\|\mathbf{X}\|_q := (\E X^q)^{1/q}$.)
\end{claim}
\begin{proof}
Indeed,
\[
\|\inprod{\tilde{v},w}\|_{q}^{q}~=~\sum_{i}|v_{i}|^{p}\left(|w_{i}||v_{i}|^{1-p}\right)^{q}~=~\sum_{i}|w_{i}|^{q}|v_{i}|^{p+q-pq}~=~\sum|w_{i}|^{q}~\leq~1~.
\]
\end{proof}
Consider now random variables $Z_j = \inprod{\tilde{v}^{(j)}, w} - \inprod{v, w}$. We have $\E Z_j = 0$, and by the above claim together with the triangle inequality $\|Z_j\|_q \leq 2$. Now, if $q\geq 2$, this implies that $\|Z_j\|_2 \leq 2$, and therefore $\|\inprod{\phi(v), w} - \inprod{v, w}\|_2 = \|\frac{1}{s} \sum_{j \leq s} Z_j\|_2 \leq 2 s^{-1/2} \leq \varepsilon/3$. \atodo{This is \Cref{lem:type}? because then there is missing asymptotic notation.}The Chebyshev inequality now implies that with probability $1/3$ we have $|\inprod{\phi(v), w} - \inprod{v, w}| \leq \varepsilon$.

Similarly, if $1< q < 2$, we can apply Lemma~\ref{lem:type}, to deduce that $\|\inprod{\phi(v), w} - \inprod{v, w}\|_q = \|\frac{1}{s} \sum Z_j\|_q \lesssim s^{1/q - 1} = s^{- 1/p}$, therefore if $s = \Theta(\delta^{-1} \varepsilon^{-p})$, we get $\|\inprod{\phi(v), w} - \inprod{v, w}\|_q \leq \varepsilon/3$, and again by Chebyshev, with probability $2/3$ we have $|\inprod{\phi(v), w} - \inprod{v, w}| < \varepsilon$.

The upper bound on $\coR_{\varepsilon, \delta}(\IP_{\ell_p})$ follows by \Cref{prop:sparsification-implies-cc}.
\end{proof}
We can now prove \Cref{thm:ell-p-cc} stating that $\cR_{\varepsilon, 1/3}(\IP_{\ell_p}) \leq \Oh(\frac{\log n}{\varepsilon^2})$.
\begin{proof}[Proof of \Cref{thm:ell-p-cc}]
If $\frac{1}{p} + \frac{1}{q} = 1$, by symmetry between Alice and Bob we have $\cR_{\varepsilon, 1/3}(\IP_{\ell_p}) = \cR_{\varepsilon, 1/3}(\IP_{\ell_q})$, and since for any norm $\cR_{\varepsilon, 1/3}(\IP_{N}) \leq \coR_{\varepsilon, 1/3}(\IP_{N})$, we have 
\begin{equation*}
    \cR_{\varepsilon, 1/3}(\IP_{\ell_p}) \leq \min(\coR_{\varepsilon, 1/3}(\IP_{\ell_p}), \coR_{\varepsilon, 1/3}(\IP_{\ell_q})).
\end{equation*}
Since for dual exponents $p,q$, we have $\min(p, q) \leq 2$, the bound follows.
\end{proof}

%% file: two-way.tex
\section{Efficient protocol for arbitrary symmetric norm \label{sec:symmetric-two-way}}
The goal of this section is the proof of Theorem~\ref{thm:symmetric-ub}.

\symmetricub*

The main ingredient used by us is the embedding theorem, saying that any symmetric norm can be with low distortion embedded into an explicit, relatively simple space. Before we can state this theorem, we need to introduce a sum operation of normed spaces --- it is necessary to describe the target space in the embedding theorem.

For a sequence of normed spaces $V_1, V_2, \ldots V_d$, and a norm $h$ on $\bbR^d$, we can construct a normed space $V$ called the $h$-sum of $V_i$, denoted as
\begin{equation*}
    V = \bigoplus_{i \leq d}^h V_i,
\end{equation*}
such that the underlying vector space is $V := \bigoplus_{i \leq d} V_i$ and norm is defined by the following formula. For an arbitrary vector $v \in V$, with decomposition $v = \sum v_i$ where $v_i \in V_i$, we take
\begin{equation*}
    \|v\|_V := h(\|v_1\|_{V_1}, \|v_2\|_{V_2}, \ldots \|v_d\|_{V_d}).
\end{equation*}

We also need to define the Top-k norm $T^{(k)}$ on $\bbR^n$, given by
\begin{equation*}
    \|v\|_{T^{(k)}} := \sum_{j\leq k} |v|_{(i)}
\end{equation*}
Where $|v|_{(1)} \geq |v|_{(2)} \geq \ldots \geq |v|_{(n)}$ is a non-increasing (in the magnitude) rearrangement of entries of $v$.

With those definition in hand, the aforementioned embedding theorem can finally be stated.
\begin{theorem}[{\cite[Theorem 4.2]{ANNRW17}}]
\label{thm:symmetric-embedding}
For any symmetric norm $N$, and any $\delta > 0$, the normed space $(\bbR^n, \|\cdot\|_N)$ embedds with distortion $1+\delta$ into $\bigoplus_{i \leq t}^{l_\infty} \bigoplus_{k\leq n}^{\ell_1} T^{(k)}$, where $t = n^{\Oh(\delta^{-1} \log \delta^{-1})}$.
\end{theorem}

We will use this Theorem~\ref{thm:symmetric-embedding} together with Proposition~\ref{prop:embedding-preserves-cc} to prove Theorem~\ref{thm:symmetric-ub}. To this end, all we need to do is to show an explicit protocol for the inner product problem on a space of form $\bigoplus_{i \leq t}^{l_\infty} \bigoplus_{k\leq n}^{\ell_1} T^{(k)}$.

\subsection{Compositions of norm}
We will show here that if all normed spaces $V_i$ admit an efficient protocol for the inner product problem, and $h$ admits a sparsification, then also $\oplus^h_i V_i$ has an efficient protocol for the inner product problem. 
\begin{lemma}
\label{lem:sum-duality}
If $V = \bigoplus_{i \leq d}^h V_i$, then $V^* = \bigoplus_{i \leq d}^{h^*} V_i^*$.
\end{lemma}
\begin{proof}
Consider $w \in V^*$, such that $w = w_1 + \ldots + w_d$.
We have 
\begin{align*}
    \|w\|_{V^*} & = \sup_{v \in V, \|v\|_{V} \leq 1} \inprod{v, w} \\
    & = \sup_{\alpha \in \bbR^d, \|\alpha\|_h \leq 1} \sup_{\substack{ v_1, \ldots v_k \\ \|v_i\|_{V_i} = \alpha_i}} \inprod{\sum v_i, w} \\
    & = \sup_{\alpha \in \bbR^d, \|\alpha\|_h \leq 1} \sum_i \sup_{\substack{v_i \\  \|v_i\|_{V_i} = \alpha_i}} \inprod{v_i, w_i} \\
    & = \sup_{\alpha \in \bbR^d, \|\alpha\|_h \leq 1} \sum_i \alpha_i \|w_i\|_{V_i^*} \\
    & = h^*(\|w_1\|_{V_1^*}, \ldots \|w_d\|_{V_d^*}).
\end{align*}
\end{proof}

\begin{lemma}
\label{lem:composition}
If for each $V_i$ we have $\cR_{\varepsilon, 1/3}(\IP_{V_i}) \leq D_1$, and $h$ has $(\varepsilon, \frac{1}{9}, D_2)$-sparsification, then the normed space  $V := \bigoplus_{i \leq d}^h V_i$ has $\cR_{2 \varepsilon + \gamma, 1/3}(V) \leq \Oh(D_1 D_2 \log D_2 + D_2 \log d)$.
\end{lemma}
\begin{proof}
Alice is given a vector $v = \sum_{i\leq d} v_i$, she can produce vector $q \in \bbR^d$ with $q_i = \|v_i\|_{V_i}$, satisfying $\|q\|_h \leq 1$.

She can now apply the sparsification procedure to $q$, to get a sparsified $\phi(q)$ with $\|\phi(q)\|_0 \leq D_2$. For each $i \in \supp(\phi(q))$, they simulate a protocol for $V_i$ on vectors $v_i/\|v_i\|_{V_i^*}$ and $w_i / \|w_i\|_{W_i^*}$. By repeating each such protocol $\Oh(\log D_2)$-times and taking the median of the results, they can ensure that with probability $8/9$ they arrive at estimates $\tau_i$ satisfying simultaneously for all $i \in \supp(\phi(q))$ the bounds $\tau_i = \inprod{v_i, w_i} \pm \varepsilon \|v_i\|_{V_i} \|w_i\|_{W_i}$.

Alice can now compute an estimate 
\begin{equation*}
    u := \sum_i \tau_i \phi(q)_i / \|v\|_{V_i}.
\end{equation*}
We claim that with probability $2/3$ we have $|u - \sum_i \inprod{v_i, w_i}| \leq 2\varepsilon + \gamma$.

Indeed, let us first consider the vector $\eta \in \bbR^d$ with $\eta_i := \frac{\inprod{v_i, w_i}}{\|v_i\|_{V_i}}$. Note that $|\eta_i| \leq \|w_i\|_{V_i^*}$, and therefore $\|\eta\|_{h^*} \leq 1$ (since $\|w\|_{V^*} \leq 1$).

The fact that $\phi(v)$ was a sparsification for the norm $h$ implies that except with probability $1/9$ over $\phi$, we have
\begin{equation}
    \inprod{\phi(q), \eta} = \inprod{q, \eta} \pm \gamma = \sum_i \|v_i\|_{V_i} \frac{\inprod{v_i, w_i}}{\|v_i\|_{V_i}} \pm \gamma = \inprod{v, w} \pm \gamma. \label{eq:inprod-v-w}
\end{equation}
On the other hand 
\begin{equation*}
    \inprod{\phi(q), \eta} = \sum_{i \in \supp(\phi(q))} \phi(q)_i \eta_i = \sum_{i \in \supp(\phi(q)} \inprod{v_i, w_i} \frac{\phi(q)_i}{\|v_i\|_{V_i}}
\end{equation*}
Since for $i \in \supp(\phi(q))$ we have access to $\tau_i = \inprod{v_i, w_i} \pm \varepsilon\|v_i\|_{V_i} \|w_i\|_{W_i}$, this yields
\begin{align}
    \inprod{\phi(q), \eta} & = \sum_i \tau_i \frac{\phi(q)_i}{\|v_i\|_{V_i}} \pm  \sum_i \varepsilon \phi(q)_i \|w\|_{V_i^*} \nonumber \\
    & = u + \varepsilon \sum_i \phi(q) \|w_i\|_{V_i^*}.
    \label{eq:phi-q-eta}
\end{align}

Now, using again that $\phi$ was promised to be a sparsification of $q$, and the vector $r = (\|w_i\|_{V_i^*})_{i \in [d]}$ has bounded $h^*$ norm by $1$, except with probability $1/9$ we have 
\begin{equation*}
    \sum_i \phi(q)_i \|w_i\|_{V_i^*} \leq \sum_i q_i \|w_i\|_{V_i^*} + \gamma \leq \|q\|_h \|r\|_{h^*} + \gamma \leq 1 + \gamma.
\end{equation*}
Plugging this into~\eqref{eq:phi-q-eta} yields
\begin{equation*}
    u = \inprod{\phi(q), \eta} + (1 +\gamma)\varepsilon
\end{equation*}
and combining this with~\eqref{eq:inprod-v-w} we get
\begin{equation*}
    u = \inprod{v, w} \pm (\varepsilon(1+\gamma) + \gamma).
\end{equation*}

The total failure probability is bounded by $1/3$ -- we have probability $1/9$ for any of $\tau_i$ to differ by more than $\varepsilon$ fom the desired value, probability at most $1/9$ for $\inprod{\phi(q), \eta}$ to be far from the desired value, and probability $1/9$ for $\inprod{\phi(q), \eta}$ to be far from the desired value.

Finally, the total communication is $\Oh(D_2 \log d + D_1 D_2 \log D_2)$. For each index $i \in \supp(\phi(q))$ (where $|\supp(\phi(q))| \leq D_2$), Alice has to communicate this index to Bob (paying $\log d$ bits of communication). Then for each such index they use the protocol for the norm $\|\cdot\|_{V_i}$ with communication cost $D_1$, and they repeat $\Oh(\log D_2)$ times to amplify the success probability.
\end{proof}
\subsection{Top-k norm}
We will now discuss properties of the Top-$k$ norm, in order to prove that those norms admit an efficient protocol for the inner product problem. It turns out that $\|w\|_{T_{(k)}^*} = \max(\|w\|_\infty, \|w\|_1/k)$. \jtodo{Why?}

\begin{remark}
    We have isometric embeddings $\ell_\infty^k \hookrightarrow (T^{(k)})^*$ and $\ell_{\infty}^{n/k} \hookrightarrow T^{(k)}$. In particular both 
    $\coR_{1-\varepsilon, 1/3}(T^{(\sqrt{n})}) \geq \Omega(\sqrt{n})$ and     
    $\coR_{1-\varepsilon, 1/3}((T^{(\sqrt{n})})^*) \geq \Omega(\sqrt{n})$.
\end{remark}

\begin{lemma}
\label{lem:max}
If $N_1$ and $N_2$ are norms, and $N$ is a norm given by $\|v\|_N = \max(\|v\|_{N_1}, \|v\|_{N_2})$, then $\cR_{\varepsilon, 2\delta}(\IP_N) \leq \cR_{\varepsilon, \delta}(\IP_{N_1}) + \cR_{\varepsilon, \delta}(\IP_{N_2})$.
\end{lemma}
\begin{proof}
Consider Bob with a vector $w$, such that $\|w\|_{N^*} \leq 1$. By Corrolary~\ref{cor:decomposition}, he can decompose $w = w_1 + w_2$, where $\|w_1\|_{N_1^*} + \|w_2\|_{N_2^*} \leq 1$

They can now use protocol for $N_1$ to compute $\inprod{v, w_1} \pm \|w_1\|_{N_1} \varepsilon$, since $\|v\|_{N_1} \leq \|v\|_N \leq 1$, and similarly they can use protocol for $N_2$ to compute $\inprod{v, w_2} \pm \|w_2\|_{N_2}\varepsilon$. Clearly $\inprod{v, w} = \inprod{v, w_1} + \inprod{v, w_2}$, so by adding the estimates from those two rounds, they can get the estimate for $\inprod{v,w}$ with an additive error $(\|w_1\|_{N_1} + \|w_2\|_{N_2}) \varepsilon \leq \varepsilon$. The communication cost is $D_1 + D_2$, and the failure probability is $2\delta$, where $\delta$ is a failure probability for protocol they used to approximate $\inprod{v_i, w}$.
\end{proof}
\begin{corollary}
\label{cor:top-k}
$\cR_{\varepsilon, 1/3}(\IP_{T^{(k)}}) \leq \Oh(\frac{\log n}{\varepsilon^2})$.
\end{corollary}
\begin{proof}
Follows from Lemma~\ref{lem:max}, and \Cref{thm:ell-p-cc}, since $T_{(k)}^* = \max(\ell_\infty, \ell_1/k)$.
\end{proof}

\subsection{Proof of Theorem~\ref{thm:symmetric-ub}}
Using the embedding theorem (Theorem~\ref{thm:symmetric-embedding}) with $\delta=1/2$ and Proposition~\ref{prop:embedding-preserves-cc}, it is enough to bound $\cR_{\varepsilon, 1/3}(\IP_{V})$  for the space $V = \bigoplus_{i \leq t}^{\ell_\infty} \bigoplus_{k\leq n}^{\ell_1} T^{(k)}$.

By Corollary~\ref{cor:top-k}, we have $\cR(\IP_{T^{(k)}}) \leq \Oh(\frac{\log n}{\varepsilon^2})$. Applying Lemma~\ref{lem:composition}, we can deduce that a space $V' =\bigoplus_{k \leq n}^{\ell_1} T^{(k)}$ has $\cR(\IP_{V'}) = \Oh(\frac{\log n}{\varepsilon^4} \log \varepsilon^{-1})$. Finally, applying once again Lemma~\ref{lem:composition} together with Lemma~\ref{lem:sum-duality} and \Cref{thm:sparsification} for $\ell_1$, we get $\cR(\IP_{V}) = \Oh(\frac{\log n}{\varepsilon^{6}} \log^2 \varepsilon^{-1})$.

%% file: sparsification.tex
\section{Sparsification for symmetric norms}
In the sequel we will consider a normed space $(\bbR^n, \|\cdot\|_N)$ with a symmetric norm $N$, such that $\ell_{\infty}^k \not\hookrightarrow^{1/\varepsilon} X$, and define $p := \frac{\log k}{\log \varepsilon^{-1}}$, such that $(\frac{1}{\varepsilon})^p = k$. Note that by \Cref{prop:embedding-lb} if $\ell_\infty^k \hookrightarrow^{1/\varepsilon} X$, then $\coR_{\varepsilon, 1/3}(\IP_{X}) \geq \Omega(k)$. We wish to show the converse; ideally, a statement of form if $\ell_{\infty}^k \not\hookrightarrow^{1/\varepsilon} X$, then $X$ admits $(\varepsilon, 1/3, \Oh(k))$-sparsification (see Definition~\ref{def:sparsification}), or to put it differently $(\varepsilon, 1/3, \Oh(\frac{1}{\varepsilon^p}))$-sparsification --- this would be a direct generalisation of Theorem~\ref{thm:sparsification}.

We will in fact show only quantitatively weaker version, specifically the theorem mentioned in the introduction.
\symmetricsparsification*

More concretely, we will prove the following slightly stronger statement.
\begin{theorem}
\label{thm:symmetric-sparsification-stronger}
If $\ell_\infty^k \not\hookrightarrow^{1/\varepsilon} (\bbR^n, \|\cdot\|_N)$ for a symmetric norm $N$, this norm has $(\varepsilon, 1/3, \tau)$-sparsification, where $\tau = (p/\varepsilon)^{4p} \log^2 (n))$ and $p = \frac{\log k}{\log \varepsilon^{-1}}$.
\end{theorem}

In particular, when $\varepsilon < \frac{1}{\log n}$ this yields sparsity $k^{\Oh(1)}$.

The proof of this theorem will follow similar outline as the proof of Theorem~\ref{thm:sparsification}, except that each step along the way is more involved, and many steps are more abstract. We recommend comparing those ideas to the ideas in more straightforward proof of Thoerem~\ref{thm:sparsification}, to have a mental picture of the overall strategy.

\subsection{Proof of Theorem~\ref{thm:symmetric-sparsification}}
Let us consider an arbitrary vector $v \in \bbR^n$, with $\|v\|_N \leq 1$, and a random variable $t \in [n]$, drawn according to some distribution $\Pr(t = i) = p_i$, which could potentially depend on the vector $v$. Let us take now a random $1$-sparse vector $\tilde{v} := e_t \frac{v_t}{p_t}$, where $e_t$ is a standard basis vector. Clearly, by construction $\E \tilde{v} = v$. 

We will specify a suitable distribution for a random variable $t$, such that we can upper bound the quantity $\sup_{w \in B_{N^*}} \| \inprod{\tilde{v}, w} \|_{q, \infty}$, where $q$ is the dual exponent to $p$ (i.e. $1/p + 1/q = 1$). If we have an upper bound $\sup_{w \in B_{N^*}} \|\inprod{\tilde{v}, w}\|_{q, \infty} \leq K$, it implies $\Oh((p K / \varepsilon)^{2p})$-sparsification for a vector $v$ --- in short because we can bound $\|\inprod{\tilde{v}, w}\|_{q'} \leq O(p) \|\inprod{\tilde{v},w}\|_{q, \infty} = O(K p)$ for a slightly smaller $q'$, and therefore by Lemma~\ref{lem:type} it is enough to sample $O(pK/\varepsilon)^{2p}$ vectors $\tilde{v}$ independently at random for an average of $\inprod{\tilde{v}_i, w}$ to concentrate within $\varepsilon$ of its expectation.

In order to implement this proof idea, let us first show a convenient bound for the quantity $\sup_{w \in B_{N^*}} \|\inprod{\tilde{v}, w}\|_{q, \infty}$.

\begin{lemma}
\label{lem:weak-moment-subsampling}
For any $v \in \bbR^n$ and the random vector $\tilde{v} := e_t v_t / p_t$ where $t \in [n]$ is a random variable distributed according to $\Pr(t = i) = p_i$, and any $q \geq 1$, if $p_i = 0$ for all $i$ with $v_i = 0$, we have 
\begin{equation*}
    \sup_{w \in B_{N^*}} \|\inprod{\tilde{v}, w}\|_{q, \infty} \leq \sup_{S \subset [n]} \Pr(t \in S)^{-1/p} \|v_S\|_{N}~,
\end{equation*}
where $v_S$ is a restriction of $v$ to coordinates in $S$, and $p, q$ are dual exponents $(\frac{1}{p} + \frac{1}{q} = 1)$.
\end{lemma}
\begin{proof}
Expanding the definition of $\|\inprod{\tilde{v}, w}\|_{q, \infty}$, we get
\begin{align*}
    \|\inprod{\tilde{v}, w}\|_{q, \infty} & = \sup_{\lambda} \Pr(|\inprod{\tilde{v}, w}| \geq \lambda)^{1/q} \lambda \\
    & = \sup_{S \subset [n]} \Pr(t \in S)^{1/q} (\min_{i \in S} \frac{|v_i w_i|}{p_i}).
\end{align*}
We could now take $\sup_{w \in B_{N^*}}$, to see that
\begin{equation}
        \sup_{w \in B_{N^*}} \|\inprod{\tilde{v}, w}\|_{q, \infty} =
    \sup_{S \subset [n]} \Pr(t\in S)^{1/q} \sup_{w \in B_{N^*}} \min_{i \in S} \frac{|v_i w_i|}{p_i}. \label{eq:tmp-first}
\end{equation}
For any $S$, let $w_S$ be a vector with $(w_S)_i = \frac{p_i}{v_i}$ for $i \in S$, and $(w_S)_i = 0$ otherwise. When the set $S$ is fixed, the supremum is achieved for $w = w_S/\|w_S\|_{N^*}$, that is
\begin{equation}
\sup_{w \in B_{N^*}} \min_{i \in S} \frac{|v_i w_i|}{p_i} = \frac{1}{\|w_S\|_{N^*}}.
\label{eq:tmp-second}
\end{equation}
Indeed, if we had any $w$ such that at least one of the terms $\frac{|v_i w_i|}{p_i}$ was greater than the minimum, we could decrease the corresponding $|w_i|$, preserving the property $w \in B_{N^*}$, and without affecting $\min_i \frac{|v_i w_i|}{p_i}$. Hence, there is an optimal vector $w$ for which all terms $\frac{|v_i w_i|}{p_i}$ for $i \in S$ are equal to each other --- that is all $|w_i|$ for $i \in S$ are proportional to $\frac{p_i}{|v_i|}$. Among those, the one maximizing $\frac{|v_i w_i|}{p_i}$ is the one with maximal $N^*$ norm, i.e. the vector $\frac{w_S}{\|w_{S}\|_{N^*}}$.

Moreover, by the definition of the dual norm, we have  $\inprod{w_S, v_S} \leq \|w_S\|_{N^*} \|v_S\|_N$, or to put it differently
\begin{equation}
    \|w_S\|_{N^*} \geq \|v_S\|_N^{-1} \inprod{w_S, v_S} = \|v_S\|_{N}^{-1} (\sum_{i \in S} p_i) = \|v_S\|_{N}^{-1} \Pr(t \in S).
   \label{eq:tmp-third}
\end{equation}
Plugging together inequalities \eqref{eq:tmp-first}, \eqref{eq:tmp-second} and \eqref{eq:tmp-third}, we get
\begin{equation*}
    \sup_{w \in B_{N^*}} \|\inprod{\tilde{v}, w}\|_{q, \infty} \leq
    \sup_{S\subset [n]} \Pr(t \in S)^{1/q} \Pr(t \in S)^{-1} \|v_S\|_N = \sup_{S\subset[n]} \Pr(t\in S)^{-1/p} \|v_S\|_N.
\end{equation*}
\end{proof}

We will leverage the assumption that $\ell_\infty^k \not\hookrightarrow^{1/\varepsilon} (\bbR^n, \|\cdot\|_N)$ through the following property.
\begin{lemma}
\label{lem:disjoint-sum}
If $\ell_\infty^k \not\hookrightarrow^{1/\varepsilon} (\bbR^n, \|\cdot\|_N)$, 
and if $v_1, v_2, \ldots v_k$ are vectors with disjoint support, then
\begin{equation*}
    \|\sum_{i \in [k]} v_i\|_N > \varepsilon^{-1} \min_i \|v_i\|_N.
\end{equation*}
\end{lemma}
\begin{proof}[Proof sketch.]
We can assume without loss of generality that $\|v_i\|_N = \|v_j\|_N$ for all $i,j \in [k]$. If the inequality 
    $\|\sum v_i\|_N > \varepsilon^{-1} \min_i \|v_i\|_N$ was violated then mapping $e_i \mapsto v_k$ would give an embedding of $\ell_\infty^k$ into $(\bbR^n, \|\cdot\|_N)$ with distortion at most $\varepsilon^{-1}$.
\end{proof}
\jtodo{Expand this proof}

Consider now a collection of disjoint sets $T_1, T_2, \ldots T_R$, vectors and a vector $v$ with $\|v\|_N \leq 1$ such that $v_{T_i} = \alpha_i \mathbf{1}_{T_i}$, where $v_{T_i}$ is a restriction of the vector $v$ to coordinates in the set $T_i$., Clearly, we have $v = \sum_{i} v_{T_i}$. We will pick probabilities $p_i = \frac{1}{R} \frac{1}{|T_j|}$ for $i \in T_j$ (and similarly for $T_*$). That is, in order to sample coordinate $t \in [n]$, we first sample uniformly one of the sets $T_i$ or $T_*$, and then a uniformly random coordinate from this set. As before, let $\tilde{v}$ be a random vector defined as $\tilde{v} = e_t \frac{v_t}{p_t}$ with probability $p_t$.
\begin{lemma}
\label{lem:general-sparsification}
For a vector $v$ and a random index $t$ of the form described above, we have
\begin{align}
    \sup_{S \subset [n]} \Pr(t \in S)^{1/p} \|v_S\|_N \leq (k R)^{1/p},
    \label{eq:sup-bound}
\end{align}
where $\ell_{\infty}^k \not\hookrightarrow (\bbR^n, \|\cdot\|_N)$ and $p := \frac{\log k}{\log \varepsilon^{-1}}$.
\end{lemma}
\begin{proof}
We will show that for any set $S$ either the desired inequality~\eqref{eq:sup-bound} holds, or we can find another set $S'$ with $|S| < |S'|$ such that 
\begin{align*}
    \Pr(t \in S)^{1/p} \|v_S\|_N \leq \Pr(t \in S')^{1/p} \|v_{S'}\|_N.
\end{align*}
By the mathematical induction on $n - |S|$ this is enough to prove the lemma.

Indeed, take $S_i = S \cap T_i$, and consider two cases. Either there is some $i_0$ such that $|S_{i_0}|/|T_{i_0}| > \frac{1}{k}$, or for all $i$ we have $|S_{i_0}| / |T_{i_0}| \leq \frac{1}{k}$. 

In the former case, we can bound $\|v_S\|_N \leq \|v\|_N = 1$, and
\begin{align*}
    \Pr(i \in S) = \frac{1}{R} \sum_{i} \frac{|S_i|}{|T_i|} \geq \frac{1}{R} \frac{|S_{i_0}|}{|T_{i_0}|} > (k R)^{-1}.
\end{align*}
hence $\Pr(i \in S)^{-1/p} \|v_S\|_N \leq (kR)^{1/p}$ as desired.

On the other hand if for all $i$ we have $\frac{|S_i|}{|T_i|} \leq \frac{1}{k}$, we can find a collection of disjoint sets $S^{(1)} = S, S^{(2)}, S^{(3)}, \ldots S^{(k)}$ such that for all $i,j$ we have $|S^{(i)} \cap T_j| = |S_j|$. Let us define $S' = \bigcup_{i} S^{(i)}$, so that $v_{S'} = \sum_i v_{S^{(i)}}$. We can now apply Lemma~\ref{lem:disjoint-sum} to vectors $v_{S^{(i)}}$ (they all have the same norm, since they are all permutations of each other), to deduce
\begin{equation*}
	\|v_{S'}\|_N=\|\sum_{i \in [k]} v_{S^{(i)}}\|_N > \varepsilon^{-1} \min_i \|v_{S^{(i)}}\|_N=\varepsilon^{-1} \|v_{S}\|_N
\end{equation*}

On the other hand $\Pr(i \in S') = \frac{1}{R} \sum_{i \in R} \frac{|S' \cap T_i|}{|T_i|} = k \Pr(i \in S)$. Hence,
\begin{equation*}
\Pr(i \in S)^{-1/p} \|v_S\|_N \leq k^{1/p} \varepsilon \Pr(i \in S')^{-1/p} \|v_{S'}\|_N = \Pr(i \in S')^{-1/p} \|v_{S'}\|_N.
\end{equation*}
\end{proof}

We will use this bound for structured vectors to show a similiar upper bound for arbitrary vectors with $R = \Oh(\log n)$, essentially by rounding each coordinate to the nearest value $2^{-j}$.
\begin{lemma}
\label{lem:weak-moment-bound}
For any vector $v$ with $\|v\|_N \leq 1$, there is a distribution $(p_1,p_2, \ldots p_n)$ of a random variable $t \in [n]$ such that the random vector $\tilde{v} := e_t \frac{v_t}{p_t}$ satisfies
\begin{equation*}
    \sup_{w \in B_{N^*}} \|\inprod{\tilde{v}, w}\|_{q, \infty} \leq \Oh(k^{1/p} \log^{1/p} (n)).
\end{equation*}
\end{lemma}
\begin{proof}
Consider sets $T_i = \{j : 2^{-i} < v_j \leq 2^{-i + 1} \}$ for $i \leq R$ where $R := 3 \log_2(n)$ and $T_+ = \{j : v_j \leq 2^{-R} \}$. In particular each coordinate $j \in T_{+}$ satisfies $v_j \leq \frac{1}{n^3}$, and therefore $\|v_{T_+}\|_N \leq \frac{1}{n^2}$.

We will set probabilities $p_j = \frac{1}{R+1} \frac{1}{|T_i|}$ for $j \in T_i$ (and similarly $\frac{1}{R+1} \frac{1}{|T_+|}$ for $j \in T_+$). According to Lemma~\ref{lem:weak-moment-subsampling}, it is enough to bound 
\begin{equation*}
    \sup_{S \subset [n]} \Pr(t \in S)^{-1/p} \|v_S\|_N.
\end{equation*}

Consider now a vector $v' := \sum_i 2^{-i} \mathbf{1}_{T_i}$, and a random variable $t'$ where $\Pr(t' = j) = \frac{1}{R} \frac{1}{|T_i|}$ for $j \in T_i$ (and $\Pr(t' \in T_+) = 0$). We can apply Lemma~\ref{lem:general-sparsification} to the pair $(v', t')$ and deduce
\begin{equation*}
    \sup_{S \subset [n]} \Pr(t' \in S)^{1/p} \|v'_S\|_N \leq (kR)^{1/p}.
\end{equation*}

Consider now arbitrary $S$ and decompose $S = S' \cup S''$ where $S'' = S \cap T_+$ and $S' = S - S''$, we have
\begin{equation}
\label{eq:split}
    \Pr(t \in S)^{-1/p} \|v_S\|_N \leq \Pr(t \in S)^{-1/p} \|v_{S'}\|_N + \Pr(t\in S)^{-1/p} \|v_{S''}\|_N
\end{equation}

We can bound those two terms separately. On one hand $\Pr(t \in S) \geq \Pr(t \in S') = \frac{R}{R+1} \Pr(t' \in S')$ and $\|v_{S'}\|_N \leq 2 \|v'_{S'}\|_N$ (since we have coordinate-wise inequalities $|v_j| \leq 2 |v'_j|$ for $j \in S'$), therefore
\begin{align}
    \Pr(t\in S)^{-1/p} \|v_{S'}\|_N & \leq 2 \left(\frac{R+1}{R}\right)^{1/p} \Pr(t' \in S')^{-1/p} \|v'_{S'}\|_N  \nonumber \\
    & \leq 2 \left(\frac{R+1}{R}\right)^{1/p} (kR)^{1/p} = 2 (k (R+1))^{1/p}. \label{eq:split-first}
\end{align}
On the other hand
\begin{equation}
    \Pr(t \in S)^{-1/p} \|v_{S''}\|_N \leq (\frac{1}{R} \frac{1}{|T_+|})^{-1/p} \|v_{T_+}\|_N \lesssim (n \log n)^{1/p} \frac{1}{n^2} = o(1). \label{eq:split-second}
\end{equation}
Combining now inequalities \eqref{eq:split}, \eqref{eq:split-first} and \eqref{eq:split-second}, get the desired result
\begin{equation*}
    \sup_{w \in B_{N^*}} \| \inprod{\tilde{v}, w}\|_{q,\infty} \leq \Oh(k^{1/p} \log^{1/p} n).
\end{equation*}
\end{proof}
With this lemma in hand, we are ready to prove the Theorem~\ref{thm:symmetric-sparsification-stronger}.
\begin{proof}[Proof of Theorem~\ref{thm:symmetric-sparsification-stronger}]
For a vector $v \in B_{N}$, let us consider a random vector $\tilde{v}$ distributed as in Lemma~\ref{lem:weak-moment-bound}, and $s$ independent copies of $\tilde{v}$, say $\tilde{v}_1, \ldots \tilde{v}_s$, for some $s$ that will be specified later.

Take $\phi(v) := \frac{1}{s}\sum_{i \leq s} \tilde{v}_i$. Since $\|\tilde{v}_i\|_0 = 1$, we have $\|\phi(v)\|_0 \leq s$, and since $\E \tilde{v}_i = v$, we also have $\E \phi(v) = v$. We wish to show that for any given $w \in B_{N^*}$ with probability $2/3$ we have $\inprod{\phi(v), w} = \inprod{v, w} \pm \varepsilon$.

Let us define $Z_i = \inprod{\tilde{v}_i, w} - \inprod{v, w}$, moreover let us take $p' = 2p$, and $q'$ a dual exponent (i.e. satisfying $1/q' + 1/p' = 1$).

By Lemma~\ref{lem:weak-markov-inverse}, we have $\|Z_i\|_{q'} \lesssim p \|Z_i\|_{q, \infty}$, and by Lemma~\ref{lem:weak-moment-bound} we have $\|Z_i\|_{q, \infty} \leq \|\inprod{\tilde{v}, w}\|_{q,\infty} + 1 \leq \Oh(k^{1/p} \log^{1/p} n)$.

Applying now Lemma~\ref{lem:type}, we have $\|\frac{1}{s} \sum_{i \leq s} Z_i\|_{q'} \lesssim s^{-1/p'} \|Z_1\|_{q'} \lesssim p \left(\frac{k \log n}{\sqrt{s}}\right)^{1/p}$.

If we now chose $s = \left(\frac{C p}{\varepsilon}\right)^{2p} (k \log n)^2$ for some universal constant $C$, we can ensure that $\|\frac{1}{s} \sum_{i \leq s} Z_i\|_{q'} \leq \varepsilon/5$, and by Markov inequality with probability $2/3$ we have $|\frac{1}{s} \sum_{i \leq s} Z_i| \leq \varepsilon$, or equivalently
\begin{equation*}
    \inprod{\phi(v), w} = \inprod{v, w} \pm \varepsilon,
\end{equation*}
as desired.
\end{proof}

\begin{remark}
By the way, we can write $(C p/\varepsilon)^{4p} \log^2 n = k^{\Oh(\log \log k)} + \log^3 n$, if we consider this bound to be more aesthetically pleasing.
\end{remark}

%% file: extension-complexity-connection.tex
\section{Extension complexity and communication complexity of $\IP_P$ \label{sec:extension-complexity}}
In this section we prove Proposition~\ref{prop:xc-weak} and Theorem~\ref{thm:xc-strong}.

As a reminder, in this section we consider an arbitrary origin-symmetric polytope $P$ (with $P=-P$) and the associated norm $\|\cdot\|_P$ on $\bbR^n$ for which $P$ is a unit ball. We wish to connect extension complexity of the polytope $P$ with the communication complexity of the related $\IP_P$ problem.

We will first with a direct argument for $\cR_{\varepsilon, \delta}(\IP_P) \leq \Oh(\frac{\log \xc(P)}{\varepsilon^2})$ (\Cref{prop:xc-weak}) --- more concretely, we wish to show that if a $P$ is a projection of some polytope $Q$ such that $Q$ is defined by $M$ inequalities, then $\cR_{\varepsilon, \frac{1}{3}}(\IP_P) \leq \Oh(\frac{\log M}{\varepsilon^2}).$

\begin{lemma}
\label{lem:projection}
Let $P$ be an origin-symmetric polytope in $\bbR^n$. If $P$ is a projection of $Q$
then $\cR_{\varepsilon, \delta}(\IP_{P}) \leq \cR_{\varepsilon, \delta}(\IP_{Q})$.
\end{lemma}
\begin{proof}
Consider $Q \in \bbR^{n}$ and a subspace $U \subset\bbR^n$ such that $P$ is the image of $Q$ for the projection $\pi : \bbR^n \to \bbR^n / U^\perp$.

By \Cref{lem:hahn-banach}, the norm $\|\cdot\|_P$ on $\bbR^n / U^{\perp}$ is isometrically isomorphic with the dual to the restriction of $\|\cdot\|_{Q^*}$ to the subspace $U$. Therefore we have 
\begin{equation*}
    \cR_{\varepsilon, \delta}(\IP_P) = \cR_{\varepsilon,\delta}(\IP_{P^*}) = \cR_{\varepsilon, \delta}(\IP_{Q^*|_{U}}) \leq \cR_{\varepsilon, \delta}(\IP_{Q^*}) = \cR_{\varepsilon, \delta}(\IP_Q),
\end{equation*}
where the inequality follows from \Cref{prop:embedding-preserves-cc}.
\end{proof}

\begin{lemma}
\label{lem:vertices-count}
If $P$ is an origin-symmetric convex polytope defined by $M$ inequalities, then $P^*$ has at most $M$ vertices.
\end{lemma}
\begin{proof}
This follows from \Cref{lem:dual-intersection} --- for a polytope $P$ defined by $P = \{x : \forall i -1 \leq A_i x \leq 1\}$ we have $P = \bigcap_i \{x : -1 \leq \inprod{A_i, x} \leq 1\}$, and therefore $P^* = \conv(\bigcup_i K_i^*)$, where $K_i := \{x : -1 \leq \inprod{A_i, x} \leq 1\}$. It is easy to check that $K_i^* = \conv(\{A_i, -A_i\})$, and therefore $P^* = \conv(\{\pm A_i\}_{i})$.
\end{proof}

\begin{lemma}
\label{lem:few-vertices-ub}
If $P$ is an origin-symmetric polytope with $M$ vertices, then $\coR_{\varepsilon, 1/3}(\IP_{P}) \leq \Oh(\frac{\log M}{\varepsilon^2})$.
\end{lemma}
\begin{proof}
In the communication problem Alice has arbitrary point $v \in \bbR^n$ with $\|v\|_P \leq 1$, which is just a point $v \in P$. She can express it as a convex combination $v = \sum_i \lambda_i v_i$ where all $v_i$ are vertices of the polytope $P$ and $\sum \lambda_i = 1$. Now consider a distribution $\mathcal{D}$ over vertices of $P$ specified by the vector which has $\Pr_{x \sim \mathcal{D}}(x = v_i) = \lambda_i$.

We have $\E_{x \mathcal{D}} x = v$ and therefore by linearity of expectation $\E_{x \mathcal{D}} \inprod{x, w} = \inprod{v,w}$ for Bobs vector $w$. Since $\|w\|_{P^*} \leq 1$, and all vertices of $P$ by definition have $P$-norm equal to one, the random variable $\inprod{x, w}$ satisfies $|\inprod{x,w}| \leq 1$. As ususal, by the Chebyshev inequality is enough to sample $t = \Oh(\varepsilon^{-2})$ independent copies of $x$ from the distribution $\mathcal{D}$, to ensure that the average $\frac{1}{t} \sum_{i \leq t} \inprod{x^{(i)}, w} = \inprod{v,w} \pm \varepsilon$ with probability $2/3$.

Alice can now just perform the sampling herself and send the names of those vertices to Bob --- she can specify a vertex of a polytope using $\Oh(\log M)$ bits of communication per sample.
\end{proof}
With those lemmas in hand we are ready to show \Cref{prop:xc-weak}.
\begin{proof}[Proof of \Cref{prop:xc-weak}]
If $P$ is a projection of $Q$ by \Cref{lem:projection} we have $\cR_{\varepsilon, 1/3}(\IP_{P}) \leq \cR_{\varepsilon, 1/3}(\IP_Q)$, now by symmetry between Alice and Bob we have $\cR_{\varepsilon, 1/3}(\IP_Q) = \cR_{\varepsilon, 1/3}(\IP_{Q^*})$, and since $Q$ was defined by $M$ inequalities, $Q^*$ has at most $M$ vertices (by \Cref{lem:vertices-count}). Finally, since $Q^*$ has at most $M$ vertices we can apply \Cref{lem:few-vertices-ub} to deduce $\cR_{\varepsilon, 1/3}(\IP_{Q^*}) \leq \Oh(\frac{\log M}{\varepsilon^2})$.
\end{proof}

We will now concentrate on the proof of~\Cref{thm:xc-strong}. To this end we will use the following convenient characterization of the non-negative rank of matrix $S$.

\begin{theorem}[\cite{FFGT12}]
\label{thm:nnr-cc}
For a non-negative matrix $S\in \bbR^{X \times Y}$, we say that the randomized communication protocol computes $S$ in expectation, if Alice and Bob given $x \in X$ and $y \in Y$ respectively, output a non-negative value $f(x,y)$ such that $\E f(x,y) = S_{x,y}$. Let $\ecc(S)$ be the communication of the optimal protocol computing $S$ in expectation, then
\begin{equation*}
    \ecc(S) = \lceil \log \rk^+(S) \rceil.
\end{equation*}
\end{theorem}

\begin{proof}[Proof of \Cref{thm:xc-strong}]
By Theorem~\ref{thm:nnr-cc} it is enough to show a protocol for Alice and Bob that computes in expectation a matrix $S'$ such that $\forall i,j |S_{i,j} - S'_{i,j}| \leq 2 \varepsilon$, assuming that they have an efficient protocol for $\IP_P$.

Here Alice is given a vertex $v \in P$, and Bob is given an inequality $i$ among those defining $P$ (either $\inprod{x, A_i} \leq 1$ or $\inprod{x, A_i} \geq -1$). Note that $\|A_i\|_{P^*} \leq 1$, so they can in fact use the protocol for $\IP_P$ to compute a value which is with probability $2/3$ in a range $\inprod{v, A_i} \pm \varepsilon$.

They can improve the failure probability by instantiating it independently $\Oh(\log \varepsilon^{-1})$ times and taking the median of reported answers --- this method provides an answer in range $\inprod{v, A_i}$ with probability $1-\varepsilon/2$ --- and hence they can attempt to compute slack as $1 - \inprod{v, A_i})$ (or $1 + \inprod{x, A_i}$) --- and if the answer is outside the range $[0,2]$, they just project it to this range.

This way, the answer $f(v, A_i)$ they are computing is a non-negative, random variable bounded by $2$, such that with $\Pr(|f(v, i) - S_{v, i}| > \varepsilon) < \varepsilon/2$ --- this implies that for all $v, i$ we have $\E f(v, i) = S_{v, i} \pm 2 \varepsilon$, proving 
\begin{equation*}
    \log \annr{2\varepsilon}(S_{P}) \leq \Oh(\log \varepsilon^{-1} \cR_{\varepsilon, 1/3}(\IP_P)).
\end{equation*}
\end{proof}